\documentclass[10pt, conference, twocolumn,a4paper]{IEEEtran}

\usepackage[english]{babel}
\usepackage[utf8]{inputenc}
\usepackage{algorithm}
\usepackage{algpseudocode}

\usepackage{lipsum,multicol}
\usepackage{bigdelim}
\usepackage{multirow}
\usepackage{calc}
\usepackage{enumerate}
\usepackage{cite}
\usepackage{float}
\usepackage[T1]{fontenc}
\usepackage{enumerate}
\usepackage{amstext,amsthm,latexsym,amssymb,graphicx,tabls}
\usepackage{amssymb,amsthm}
\usepackage{amsfonts,amsmath}
\usepackage{latexsym}
\usepackage[mathscr]{euscript}
\usepackage{graphics}
\usepackage{graphicx}
\usepackage{cite}
\usepackage{cases}
\usepackage{tabularx,bbm}
\usepackage{color}
\usepackage{epstopdf}

\newcommand{\beq}{\begin{equation}}
\newcommand{\eeq}{\end{equation}}
\newcommand{\beqq}{\begin{equation*}}
\newcommand{\eeqq}{\end{equation*}}
\newcommand{\ei}{\end{itemize}}
\newcommand{\bi}{\begin{itemize}}
\newcommand{\ee}{\end{enumerate}}
\newcommand{\be}{\begin{enumerate}}

\newtheorem{definition}{Definition}

\newtheorem{prop}{Proposition}

\theoremstyle{remark}

\newcommand{\E}{\mathbbm{E}}

\begin{document}
\title{Learning from Experience: A Dynamic Closed-Loop QoE Optimization for Video Adaptation and Delivery}
\author{\IEEEauthorblockN{Imen Triki$^\star$, Quanyan Zhu$^\circ$, Rachid El-Azouzi$^\star$, Majed Haddad$^\star$ and Zhiheng Xu$^\circ$}
\IEEEauthorblockA{$^\star$CERI/LIA, University of Avignon,
Avignon, France.\\
$^\circ$ NYU Tandon School of Engineering, New York, USA.}
}

\maketitle
\begin{abstract}
The quality of experience (QoE) is known to be subjective and context-dependent. Identifying and calculating the factors that affect QoE is indeed a difficult task. Recently, a lot of effort has been devoted to estimate the users' QoE in order to improve video delivery. In the literature, most of the QoE-driven optimization schemes that realize trade-offs among different quality metrics have been addressed under the assumption of homogenous populations. Nevertheless, people perceptions on a given video quality may not be the same, which makes the QoE optimization harder. This paper aims at taking a step further in order to address this limitation and meet users' profiles. To do so, we propose a closed-loop control framework based on the users' (subjective) feedbacks to learn the QoE function and optimize it at the same time. Our simulation results show that our system converges to a steady state, where the resulting QoE function noticeably improves the users' feedbacks.
\end{abstract}

\vspace{.5cm}

\section{Introduction}

Due to the emergence of smartphones in human daily life and the tremendous advances in broadband access technologies, video streaming services have greatly evolved over the last years to become one of the most provided services in the Internet. According to Cisco \cite{CiscoSurvey}, http video streaming will be $82 \%$ of all internet consumers' traffic by $2020$, up from $70\%$ in 2015. It is then well understood that more and more interest is being today accorded to video streaming services in the hope of making all people satisfied with the video delivered quality. Nevertheless, satisfying all users at once is a hard task; in fact, someone may appreciate a video quality that someone else may not appreciate at all. This makes the study of the QoE too complex.

In the literature, different studies have been explored to express the user's QoE as an explicit function of some metrics. Some works claim that the QoE can be directly mapped to some QoS metrics such as the throughput, the jitter and the packet loss \cite{Miguel16NPA, Mellouk12NOC}. Other recent works found that the QoE could be expressed through some application metrics such as the frequency of video freezing (stalls), the startup delay, the average video quality and the dynamic of the quality changing during the streaming session \cite{Vriendt14BellLabs,Balachandran12ACM}. However, the QoE may change depending on the video context and some other external factors linked to the user himself \cite{Amour2015}, which explains the trend of using new standardized subjective metrics such as the MOS (Mean Opinion Score) and the users' engagement rate \cite{DOB11,videoEngagement, Balachandran13ACM}. A direct relation between the time spent in rebuffering and the user's engagement has been studied in \cite{Balachandran13ACM}. 

In the industry, various adaptive video streaming solutions have been explored to meet the users' expectations, such as Microsoft's smooth streaming, Adobe's HTTP dynamic streaming and Apple's live streaming. All these solutions use the well-known DASH (Dynamic Adaptive Streaming over HTTP) standard. Despite the emergence of several proposals to improve the QoE, there is still no consensus across these solutions since the users' perceptions are quite different.

The main motivation of this work is to make DASH deal with the very wide heterogeneity of people. At the core, lies the idea of performing real time supervision on the users' real perceptions to permanently improve the performance of quality adaptation.
To the best of our knowledge, such a paradigm has not been yet investigated for adaptive video streaming. In the literature, we found that the users' feedbacks on the video quality delivery were only used to study the human perceptions or to validate some analytical QoE models to help predict the QoE \cite{Amour2015, Testolin14MED-HOC-NET, Vriendt14BellLabs, Bampis}.

In \cite{Bampis}, QoE prediction was performed by incorporating machine learning, users' feedbacks and some of the QoE-related features such as rebuffering and memory-effect. Following the same trend, we combine machine learning with a QoE-maximization problem in a closed-loop architecture to dynamically adapt video quality with respect to the users' feedbacks. We focus on two main ideas: maximizing the feedbacks returned by users and exploiting the knowledge of future throughput. Note that, thanks to the exploitation of Big Data in network capacity modelling and prediction, throughput estimation becomes possible today and may go up to few seconds to the future\cite{she2010}. Though, very few papers were exploiting the knowledge of future throughput variation \cite{KnowingFutureInfocom13,ImenWoWMOM16,ism/TrikiAH16}. In \cite{KnowingFutureInfocom13}, authors designed a QoE-driven optimization problem and proposed a cross-layer scheme to minimize the cost of capacity usage and to avoid video stalls under the assumption of a perfect throughput knowledge. The main shortcoming of their approach is that it is only suited for classical video streaming as it ignores important visual quality metrics related to adaptive streaming such as video resolution and bitrate distribution. Holding the same assumption, authors in \cite{ImenWoWMOM16} proposed a proactive video content delivery algorithm, called NEWCAST, to adjust the quality of adaptive video streaming over the future horizon. Their work was afterward extended in \cite{ism/TrikiAH16} to make their algorithm well suited for unperfected throughput prediction.


Our contribution in this paper is twofold:
\begin{itemize}
\item We exploit the knowledge of future throughput variations in order to solve the optimization problem addressed in \cite{Yin:2015:CAD} in a smoother and faster manner based on similar mathematical analysis than in \cite{ImenWoWMOM16}.
\item We design a closed-loop framework based on client-server interactions to learn the overall users' perceptions and to fittingly optimize the quality of the streaming. The performance of our proposed framework is obtained using Matlab and NS3 simulations under multi-user scenario.
\end{itemize}
The paper is organized as follows: In Section \ref{section1}, we formulate the single-user QoE-optimization problem. Then, in Section \ref{section2}, we discuss the strategy of the optimal solution and propose a heuristic that performs close to the optimal solution. In Section \ref{section3}, we address the multi-user case and propose a closed-loop based framework using neural networks. Then, in Section \ref{section4}, we evaluate the performance of this framework through some numerical results. Section \ref{conclusion} concludes the paper.

\section {Single-user QoE problem formulation} \label{section1}
\subsection {The video streaming model}
We model a video as a set of $S$ segments (or shunks) of equal duration in second. Each segment is composed of N frames and is stored on the streaming server at different quality representations. Each representation designs a video encoding rate (hereinafter called bit-rate). Denote by $ b_1, b_2, \dots, b_L $ the available video bit-rates where $b_i<b_j \ for \ i<j $. We suppose that all the video frames are played with a deterministic rate, e.g 25 frames per second (fps). Denote by $\lambda$ this frame rate.

For each segment, the player indicates to the server the quality needed for streaming it. Let $b^{(s)}$ be the bit-rate associated to segment $s$ and \textbf{$\mathcal{B}=\{b^{(1)}, \dots, b^{(S)} \}$} be the set of bit-rates associated to all video segments.
We assume that the video playback buffer is big enough to avoid eventual buffer overflow events. We denote by $B(t_k)$ the number of segments that the playback buffer contains at time $t_k$. At the beginning of the streaming session, a prefetching stage is introduced to avoid future buffer underflows; $T_0$ seconds of video (corresponding to $x_0$ segments) have to be completely appended to the buffer before starting playing the video.
When there are no segments in the playback buffer, the video stops and a new prefetching stage is introduced to append again $T_0$ seconds of video before pursuing the lecture. This event is, hereinafter, referred to as video stall.

In our study, we exploit the knowledge of the user's throughput over a given horizon to the future $\mathcal{H}=[t_1 \dots t_N]$. Before starting the session, we propose to set all the video segments' bit-rates, to be optimally streamed over that horizon . 
We denote by $r(t_k)$ the user's estimated throughput at time $t_k$, $k \in [1 \dots N]$ and by $b_{t_k}$ the video bit-rate scheduled to be streamed at that time.
Note that $b_{t_k}$ will only depend on the throughput variation and the set of segments' bit-rates $\mathcal{B}$. In the following we denote it by $b_{t_k}(\mathcal{B},r)$.

To model the dynamic of the playback buffer, we define two different phases:
\bi
\item the start-up/rebuffering phase: referred to as \textbf{BaW-phase}\footnote{BaW:Buffer and Wait}, where the media player only downloads the video without playing it
\item the playback phase: referred to as \textbf{BaP-phase}\footnote{BaP:Buffer and Play}, where the player downloads and plays the video at the same time.
\ei
Depending on the state of the buffer at each time of observation $t_k$, we define $2$ variables $S_{BaP}(t_k)$ and $\tau_{BaW}(t_k)$ as follows:
\begin{enumerate}
\item if the player is on a BaP-phase, $S_{BaP}(t_k)$ defines the time at which that phase has started
\item if the player is on a BaW-phase, $S_{BaP}(t_k)$ defines the time at which the {\it next} BaP-phase will start
\item if the buffer is empty, $\tau_{BaW}(t_k)$ determines the duration of the following BaW-phase
\item if the buffer is not empty, $\tau_{BaW}(t_k)$ takes zero
\end{enumerate}

which can be mathematically expressed as
\vspace{-0.0cm}
\begin{equation}\label{BaP}
\footnotesize
S_{BaP}(t_k) = max \{S_{BaP}(t_{k-1}) \ , \ \delta (B(t_k)=0)\cdot(t_k + \tau_{BaW}(t_k) )\}
\end{equation}
\vspace{-0.2cm}
\begin{equation}\label{BaW}
\small
 \tau_{BaW}(t_k)= \delta (B(t_k) = 0)\cdot \{\tau \ ; \ \sum_{t=t_k}^{t_k + \tau} \frac{\lambda \cdot r(t)}{N \cdot b_{t}(\mathcal{B},r)} = x_0 \}
\end{equation}
\vspace{-0.2cm}
where
\vspace{-0.2cm}
 \[ \small \delta(X) \ \ \left\{
\begin{array}{l l}

 1 \ \text{if} \ X=\text{is true}\\
 \\
 0 \ \text{otherwise}\\

 \end{array} \right. \]
\vspace{-0.2cm}
Hence the dynamic of the playback buffer can be written as

\begin{equation}\label{B}
\small
 B(t_k) = \{ B(t_{k-1}) + \frac{\lambda \cdot r(t_k)}{N \cdot b_{t_k}(\mathcal{B},r)} - \lambda \cdot \delta(t_{k} \ge S_{BaP}(t_k)) \}^+,
\end{equation}
where $\{x\}^+ = max \{x,0\}$ is used to ensure that the playback buffer occupancy cannot be negative.

\vspace{-0.2cm}
\begin{figure} [h!]
\begin{center}
\includegraphics[width=8cm]{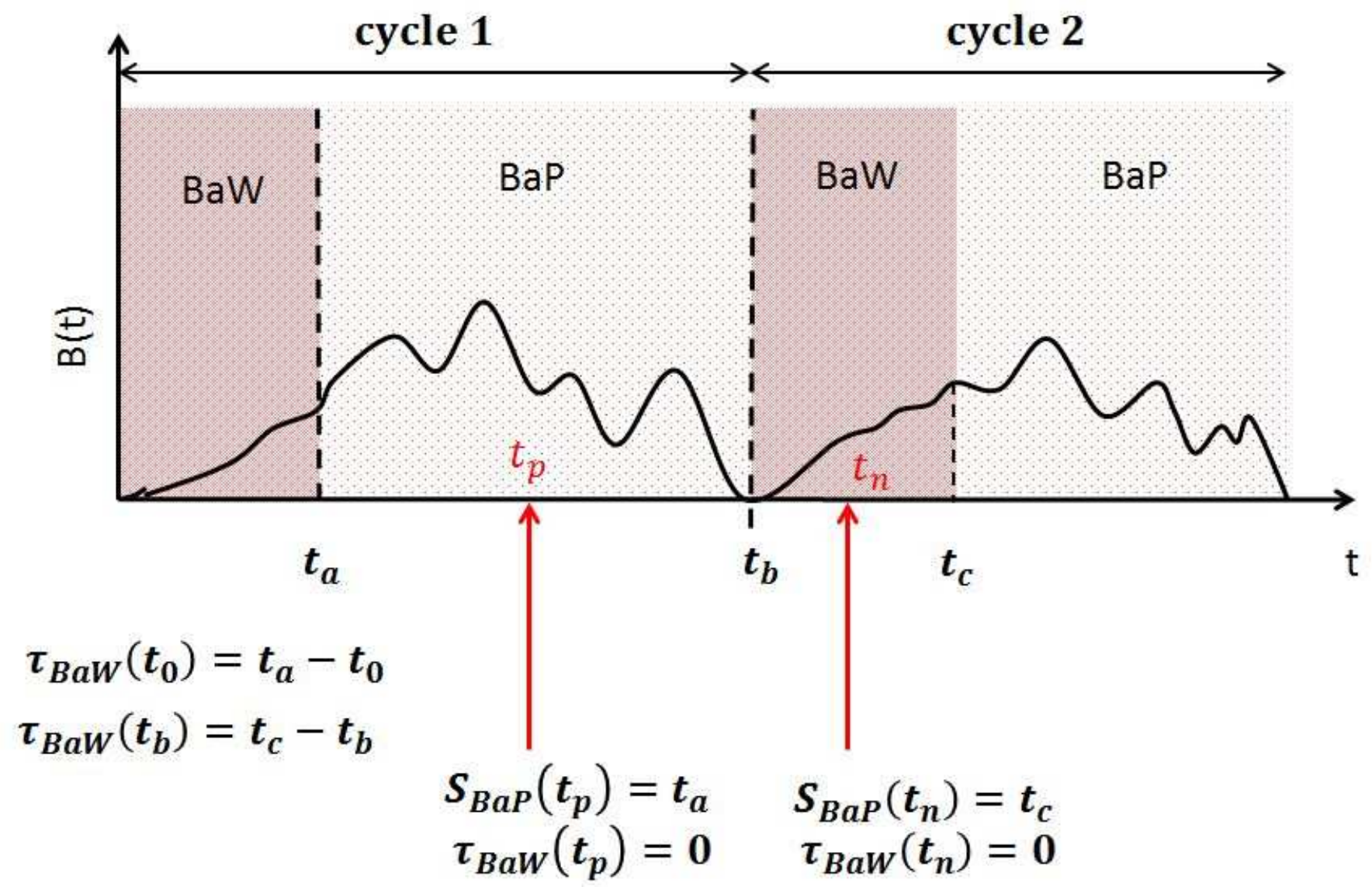}
\caption{Illustration of the playback buffer BaW-BaP cycles.}
\label{fig:NN}
\end{center}

\end{figure}

\subsection{The QoE-Optimization Problem}

The goal of bit-rate adaptation in video streaming services is to improve the users' globally perceived quality of the video. In the literature, it is challenging to quantitatively define the QoE as it encompasses many complex factors such as the user's mood, the time and the way one watches the video, the video context, etc. In this work, we use five of the most common key QoE metrics to express our objective QoE function.
\begin{enumerate}
 \item The average video quality $\phi_1$, which is the average per-segment quality over all segments given by
\vspace{-0.2cm}
\begin{equation}
 \phi_1(\mathcal{B})=\frac{1}{S} \sum_{s=1}^{S} b^{(s)}.
\end{equation}
 \item The startup delay ratio, which is the proportion of time that takes the first BaW-phase before starting the video:
\vspace{-0.2cm}
\begin{equation}
 \phi_2(\mathcal{B}) = \frac{\tau_{BaW}(t_0)}{T},
 \end{equation}
where $T$ is the video length in seconds.
 \item The average number of video quality switching given by
\vspace{-0.2cm}
\begin{equation}
\phi_3(\mathcal{B}) = \frac{1}{S-1} \sum_{s=2}^{S} \delta \{ b^{(s)} \ne b^{(s-1)}\}.
\end{equation}
 \item The number of video stalls:
\vspace{-0.2cm}
\begin{equation}
\phi_4(\mathcal{B}) = \sum_{k=1}^T \delta \{ {B(t_k) = 0} \}.
\end{equation}
 \item The rebuffering delay ratio, which is the proportion of time that take all the rebuffering events:
\vspace{-0.2cm}
\begin{equation}
\phi_5(\mathcal{B}) = \frac{1}{T} \sum_{k=1}^T { \delta \{ B(t_k) = 0 \} \cdot \tau_{BaW}(t_k)}.
\end{equation}
\end{enumerate}

As the user's preference on each of these QoE metrics may not be the same, we assign to each metric $\phi_i$ a weighting parameter $\omega_i$ to adjust its impact on the global QoE variation. As done in previous works \cite{Yin:2015:CAD}, we model our global QoE as a linear function of the weighted five aforementioned QoE metrics, namely,
\begin{equation}\label{objFunction}
 \mathcal{Q(\mathcal{B})} = \sum_{i=1}^5 \omega_i \phi_i(\mathcal{B}),
\end{equation}
where $ \omega_1 \ge 0 \ \text{and} \ \omega_i \le 0, \ \forall i \in 2, \dots, 5$.

Let $\mathcal{W} = (\omega_1, \dots,\omega_5)^\top$ be the vector of weights and $\Phi(\mathcal{B}) = (\phi_1(\mathcal{B}),\dots,\phi_5(\mathcal{B}))^\top$ be the vector of QoE metrics.

If we assume that the user tolerates at most $p$ stalls during the hole session, we end up formulating our single-user QoE optimization problem as follows
\begin{equation}\label{max}
 \underset{\mathcal{B}}{\text{max}} \
 \mathcal{Q}(\mathcal{B})= \mathcal{W}^\top \Phi(\mathcal{B}) \
 \end{equation}
 \[\textrm{s.t.} \ \ \left\{
\begin{array}{l l}

 \sum_{k=0}^{N} \frac{\lambda \cdot r(t_k)}{N \cdot {b}_{t_k}(\mathcal{B},r)} = S\\
 \\
\phi_4(\mathcal{B}) \le p \ ; \ p \in \mathbb{N};

 \end{array} \right. \]
where the first constraint ensures that the whole video will be streamed at the end of the future horizon.

\section {Proposed Solution for single-user QoE problem} \label{section2}
The QoE optimization problem defined in (\ref{max}) is a combinatorial problem with a very high complexity (NP hard). 
In \cite{Yin:2015:CAD}, authors were addressing a similar problem, but they were assuming an inaccurate throughput estimation, which justifies their choice to adopt an MPC model to solve their QoE optimization problem. The assumption of accurately knowing the future with adaptive video streaming was explored in \cite{ImenWoWMOM16} where authors proposed an \textit{ascending bitrate strategy} to optimize the video delivery. In this paper, we characterize an important propriety of the optimal strategy, which allows us to propose a heuristic approach that performs close to the optimal solution.

\subsection{Propriety of the optimal solution: Ascending bit-rate strategy per BaW-BaP cycle:}

\begin{definition}
We say a bit-rate strategy is {\bf ascending per BaW-BaP cycle}, if the quality level of the segments increases during each BaW-BaP cycle of the streaming session.
\end{definition}

\begin{prop}
Assume that there exists a solution $\mathcal{B}$ that satisfies the constraints in (\ref{max}), then there exists an ascending bit-rate per BaW-BaP cycle solution $\mathcal{B}_{as}$ that optimizes problem (\ref{max}).
\label{prop1}
\end{prop}

\begin{proof}
We shall show that for any feasible strategy $\mathcal{B}$ that satisfies the constrains in (\ref{max}), there exits an ascending bitrate per BaW-BaP cycle strategy $\mathcal{B}_{as}$ such that
$$
 \mathcal{Q}(\mathcal{B}_{as}) \geq \mathcal{Q}(\mathcal{B}).
$$
Here, we distinguish two cases:
\begin{itemize}
\item {\it Case where the session is composed of one BaW-BaP cycle, i.e, no stall during the session:}
Without loss of generality, and for the sake of illustration, we assume that we can stream and play the video in a smooth way under a non-ascending bitrate strategy $\mathcal{B}$. Then, there $ \exists \ m \le n $ such that $ b^{(m)} \ge b^{(n)}$. Let $t_p$ and $t_q$ be the requesting times of $b^{(m)}$ and $b^{(n)}$, respectively, as illustrated in Fig. \ref{fig:WW}. If we switch between qualities of segments $m$ and $n$, then the buffer state will be more relaxed toward the stall constraint since segment $m$ will be streamed in a shorter time and, then, the following segments will be appended sooner to the buffer, which may not induce buffer stalls. That said, if we reorder $\mathcal{B}$ in an ascendant way, the video will not experience any stall.
Let $\mathcal{B}_{as}$ be the resulting set after reordering $\mathcal{B}$ in an ascending way, then we have $$\phi_4(\mathcal{B}_{as})=\phi_4(\mathcal{B})=0 \ \text{and} \ \phi_5(\mathcal{B}_{as})=\phi_5(\mathcal{B})=0.$$

As we keep the same selected bitrates in $\mathcal{B}_{as}$ as in $\mathcal{B}$, the average per segment bitrate will not change, which gives
$$
\phi_1(\mathcal{B}_{as}) = \phi_1(\mathcal{B}).
$$

Since $\mathcal{B}_{as}$ is an ascending strategy, the video session will start with the lowest quality used by $\mathcal{B}$. Hence, the startup delay will be reduced by using $\mathcal{B}_{as}$ compared to $\mathcal{B}$. Therefore,
$$
\phi_2(\mathcal{B}_{as}) \leq \phi_2(\mathcal{B}).
$$
Now, let $L$ be the number of qualities selected by $\mathcal{B}$. Thus, the number of quality switching under strategy $\mathcal{B}$ will be at least $L$. On the other hand, strategy $\mathcal{B}_{as}$ will experience exactly $L-1$ quality switching since the quality level of segments increases during the session. Therefore, we have
$$
\phi_3(\mathcal{B}_{as}) \le \phi_3(\mathcal{B}).
$$

All things considered, we have
$$
 \mathcal{Q}(\mathcal{B}_{as}) \geq \mathcal{Q}(\mathcal{B}).
$$

\item {\it Case where the session is composed of more than one BaW-BaP cycle, i.e, one or more stalls during the session:}
Here, we assume that, for a given horizon window, we can stream the video under a non-ascending bitrate strategy $\mathcal{B}$ with $\phi_4$ stall events over the session ($\phi_4 \ge 1$). Undoubtedly, reordering all the segments bitrates in an ascending way will add more protection to the buffer against the stall constraint, which may reduce the number of stalls $\phi_4$. However, it does not mean that the global QoE will increase, because the duration of the stalls will change depending on the new moments of their occurrence, their corresponding requested qualities and the dynamic of the user throughput. For these reasons, our ascending bitrate strategy will not work per a hole session. In an other hand, we admit that, when a stall happens, the buffer state becomes independent of its previous states before the stall, which makes all the BaW-BaP cycles independent from each other.
Let's write $\mathcal{B}=\{ \mathcal{B}_1 , \cdots, \mathcal{B}_{\phi_4+1} \} $, where $\mathcal{B}_i$ denotes the set of bitrates used on the $i^{th}$ BaW-BaP cycle.

If we apply our previous ascending strategy on each of the ($\phi_4 + 1$) BaW-BaP cycles, we end up reducing the duration of all the BaW-phases (including the startup and the rebuffering delays) and the global number of quality switching, while maintaining the same number of stalls and the same average quality.

Let $\mathcal{B}_{as}=\{ {\mathcal{B}_1}_{as}, \cdots, {\mathcal{B}_{\phi_4+1}}_{as} \}$, then we have
 $$
 \mathcal{Q}(\mathcal{B}_{as}) \geq \mathcal{Q}(\mathcal{B}).
$$
\end{itemize}
This concludes the proof.

\end{proof}

\begin{figure} [h!]
\begin{center}
\includegraphics[width=6cm]{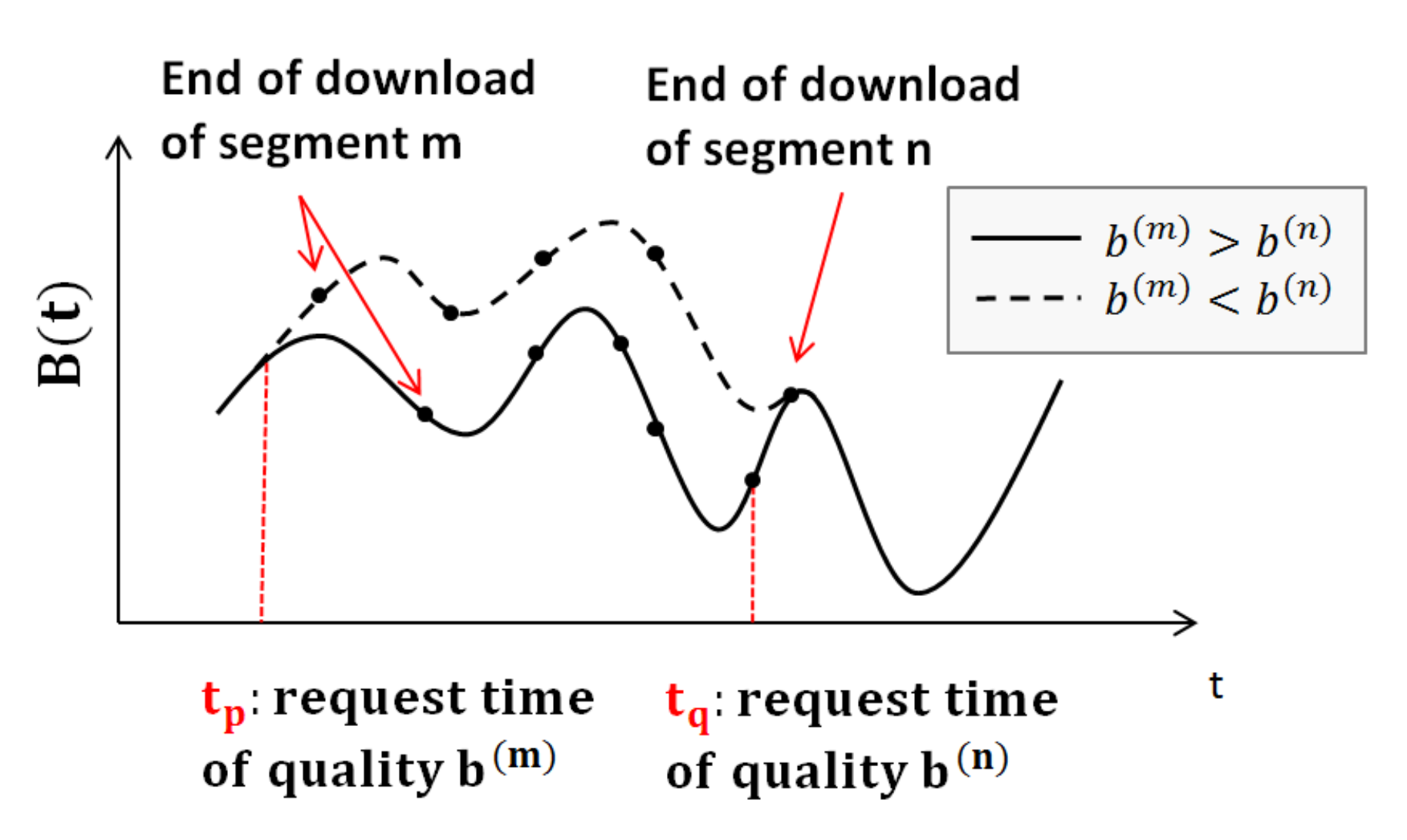}
\caption{Impact of bit-rate switching on the buffer state evolution.}
\label{fig:WW}
\end{center}

\end{figure}

\subsection{Algorithm for Optimal Solution}

In this section, we describe the main steps to follow to build an optimal solution of at most $p$ of stalls during the streaming session. The key idea of this algorithm is {\it stall enforcement}; As we assume knowing the future throughput, we are able to enforce stalls at any moment of the streaming session. Once we locate the stalls' positions (at the level of witch segment each stall should happen), we devide the session into multiple BaW-BaP cycles then look for the optimal ascending bit-rate strategy over each cycle.
The optimal number of stalls is obtained through an exhaustive research; we start computing the optimal strategy with zero stalls, then with one stall up to $p$ stalls. The stalls' distribution is also obtained through an exhaustive research. In what follows, we describe the main steps to build an optimal ascending bit-rate strategy over one BaW-BaP cycle: \\
 i) Find all the possible ascending bit-rate combinations of the BaW-phase that allow to build an ascending bit-rate strategy over the hole BaW-BaP cycle (step A and B in Fig.\ref{fig:Algo}). \\
 ii) For each BaW-phase combination, find all the possible ascending strategies that satisfy the constraints of (\ref{max}) (steps 1, 2 and 3 in Fig.\ref{fig:Algo}). \\
 iii) For each strategy, compute the QoE metrics then apply the vector of weights $\mathcal{W}$ to find the best solution. \\
 To find {\it all} the possible ascending strategies, use the tree of choice described in \cite{ImenWoWMOM16 }IV-3.

\begin{figure} [h!]
\begin{center}
\includegraphics[width=6cm]{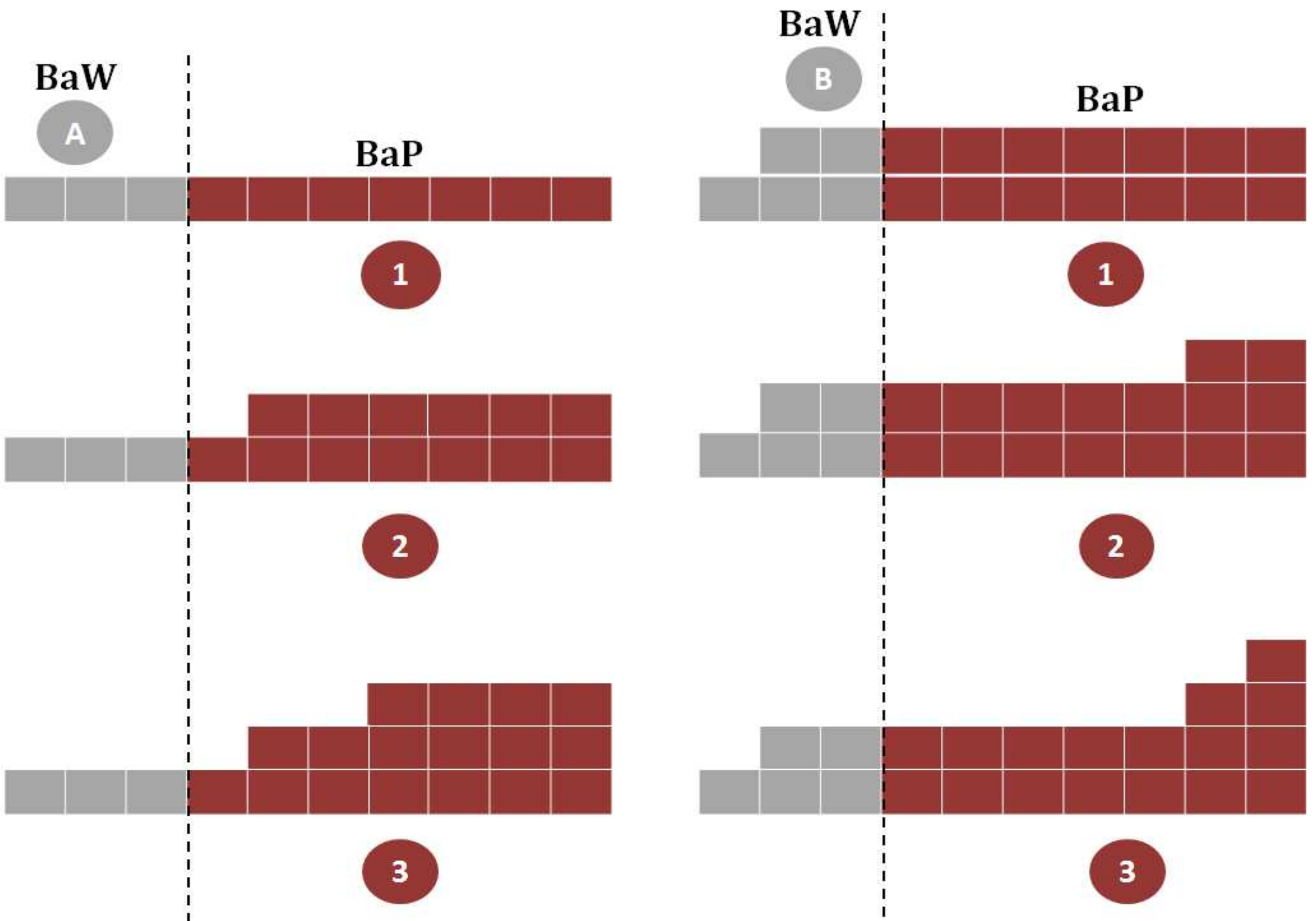}
\caption{Steps for building an increasing bit-rate strategy.}
\label{fig:Algo}
\end{center}

\end{figure}

\subsection{Heuristic for a Sub-optimal Solution}
The key idea of our heuristic is twofold: \\
1) The way we found the ascending strategy on a BaW-BaP cycle is different than the optimal strategy; Once we fix the bit-rates combination of the BaW-phase, we progressively increase the bit-rates of the BaP-phase starting by the end of the BaW-BaP cycle till reaching the point (segment) at the level of which a stall will happen if we keep increasing the quality ($5^{th}$ segment in step A-2 and $7^{th}$ segment in step A-3 of Fig.\ref{fig:Algo}). Given that the number of segments of the BaW-phase is small in general, it does not take much time to find all the possible ascending BaW-phase combinations, which makes our heuristic fast (see Algorithm \ref{alg:MAESTRO}). \\
2) rather than doing an exhaustive research on the number of stalls we process as follows: we start finding the optimal strategy with zero stalls, then, we check if the global QoE will increase with one stall enforcement (we try all the possible positions). If it does, we try to enforce a second stall, if not, we stop and return the latest strategy. We keep increasing the number of stalls as described till reaching the maximum $p$ or till the QoE function decreases.
See Algorithm (\ref{alg:CASTLE}), where $K_i$, $i\le p$ denotes the position of the $i^{th}$ stall. \\

\begin{algorithm}
\caption{MAESTRO: MAximizing qoE with aScending biTRate strategy over One-cycle}\label{alg:MAESTRO}
\begin{algorithmic}[1]
\State \textbf{Input: $\{b_l\}_{l \le L} , \ c \ , \ S \ , \ W$}
\State $M=[ \ ]$ , $b^{(s)}=b_1\ \ \forall s \in 1,\dots , x_0$ \Comment {\begin{small} {Initialization}\end{small} }
 \For{$l_{1}=1:L$} \Comment{\begin{small} {For segments of BaW-state} \end{small}
 \State $\{{b^{(s)}_{Previous}}\}_{1 \le s \le S}=\{b^{(s)}\}_{1 \le s \le S}}$
 \State $b^{(s)}=b_{l_1} \forall \ s \in x_0+1,\dots ,S$
 \State \textit{check if it is possible to stream $\{b^{(s)}\}_{1\le s \le S}$ 		 \State without stalls }
 \If {no stall happens}
 	\State $s=x_0$ \Comment {\begin{small} {Start from end to back} \end{small} }
 	\While{$s \ge 1$ \textbf{and} \textit{No stall happens}}
 		\State $b^{(s)}=b_{l_1}$
 		\State $s=s-1$
 		\State \textit{check if it is possible to stream $\{b^{(s)}\}_{1\le s \le S}$ 		 		 \State without stalls }
			\EndWhile
 	\If {\textit {a stall happens} }
 		\State $b^{(s)}= b^{(s)}_{Previous}$
 	\EndIf
 	\State $I_{l_1,l_1}= \{b^{(s)}\}_{1\le s \le S} $
 	\State \text{Compute $\Phi_{l_1,l_1}=(\phi_1, \phi2 , \phi3)$ }

 	\For {$l_2=l_1+1:L$} \Comment{\begin{small} { For segments of BaP-state}\end{small} }
 		\State $\{{b^{(s)}_{Previous}}\}_{1 \le s \le S}=\{b^{(s)}\}_{1 \le s \le S}$
 		\State $s=S$ ; \Comment{ \begin{small} {Start from end to back}\end{small} }
 		\While{$s > x_0$ \textbf{and} \textit{No stall happens}}
 			\State $b^{(s)}=b_{l_2}$
 			\State $s=s-1$
 			\State \textit{ check whether it is possible to stream 		 		
 \State $\{b^{(s)}\}_{1\le s \le S}$ without stalls }
				\EndWhile
 		\If {\textit {a stall happens} }
 			\State $b^{(s)}= b^{(s)}_{Previous}$
 		\EndIf
 		\State $I_{l_1,l_2}= \{b^{(s)}\}_{1\le s \le S} $
 		\State \text{Compute $\Phi_{l_1,l_2}=(\phi_1, \phi_2 , \phi_3)$ }
 \State $M= [M ;\Phi_{l_1,l_2} ]$

 \EndFor

 \EndIf
 \EndFor

\State \textbf{return} $ \Phi^*=\underset{M(i,:), i \ge 1} {\textbf {Argmax}} (M[w_1, w_2, w_3]^T)$
\end{algorithmic}
\end{algorithm}

\begin{algorithm}
\caption{CASTLE: asCending bitrAte STrategy over muLti-cycle sEssion }\label{alg:CASTLE}
\begin{algorithmic}[1]
\State \textbf{Input: $\{b_s\}_{s \le L} , \ c \ , \ x_0 , \ S , \ p , \ W$}
\State $\text{Bound}_{Inf}=x_0 \ , \ \text{Bound}_{Sup}=S \ $, \ i=1
\State \text{Previous QoE}= QoE without stalls
\State $\text{Previous} \ \Phi= \Phi$ without stalls
\For {$i \le p$}
\For { $ K_i \in \{x_0+1,.., S-x_0\}$ }
\If { \text{there are some stalls already positioned}}
\State $\text{Bound}_{Inf}=\text{max}\{K_j; K_j<K_i\}_{j\le i} \ \text{or} \ x_0 $
\State $\text{Bound}_{Sup}=\text{min}\{K_j; K_j>K_i\}_{j\le i} \ \text{or} \ S$
\EndIf
\State $\text{BaW-BaP}_{PreStall}= \{\text{Bound}_{Inf} .. K_i-1 \}$
\State $\text{BaW-BaP}_{PostStall} =\{K_i .. \text{Bound}_{Sup} \}$
\State $ \textbf{MAESTRO}([w_1, w_2, w_3]^T,\text{BaW-BaP}_{PreStall})$
\State $ \textbf{MAESTRO}([w_1, w_2, w_5]^T,\text{BaW-BaP}_{PostStall}) $
\State compute $\Phi_{K_i} = (\phi_1 , \phi_2 , \phi_3 , i , \phi_5) $
\State $M=[M; \Phi_{K_i}]$
\EndFor
\If { ${\text{max} \{\{MW^T \}_{j\ge 1}\} }$ > \text{Previous QoE} }
\State $K_i$=$\underset{j\ge 1}{\text{argmax}}\{\{MW^T \}_{j\ge 1}\}$
\State \text{Previous QoE= Resulting QoE}
\State $\text{Previous} \ \Phi=\Phi_{K_i}$
\State i=i+1
\Else \Comment {\begin{small}{No need for stalls to increase the QoE}\end{small}}
\State $\textbf{return} \ \Phi^*=\text{Previous} \ \Phi $
\EndIf
\EndFor
\State $\textbf{return} \ \Phi^*=\text{Previous} \ \Phi $
\end{algorithmic}
\end{algorithm}

\section{Multi-user QoE optimization problem} \label{section3}
\subsection{Problem formulation}
In this section, we extend the QoE optimization problem to the multi-user case. We propose to find the vector of weights $\mathcal{W}^*$ that maximizes the QoE among all users.
The main objective is to maximize  the users' feedbacks on the video delivery using a synthetic QoE dataset. The QoE problem of the multi-user case can be mathematically expressed as
\begin{equation}
\label{argmax}
\mathcal{W}^* \in \underset{\mathcal{W}}{\text{argmax}} \left\{\sum_{u=1}^U\E_{r_u} \{\mathcal{F}_{r_u}(\mathcal{W})\} \right\}
\end{equation}
where $r_u$ is the throughput of user $u$ and $\mathcal{F}_{r_u}(\mathcal{W})$ is his feedback on the quality he received after QoE optimization (\ref{max}) using vector $\mathcal{W}$.

\subsection{Practical solution: Closed-loop- based framework with users' feedbacks} \label{framework}
\subsubsection{Framework design}
The multi-user QoE optimization problem requires to solve problem (\ref{max}) for each user $u \in \{1,\dots,U\}$, knowing the exact value of vector $\mathcal{W}^*$ that meets all the users' preferences.  The challenge is then to combine single user QoE optimization with a QoE training mechanism in a closed-loop manner to progressively learn the value of $\mathcal{W}^*$. To do so, we develop two sub-frameworks  and make them interact together within a closed-loop based framework, one is for QoE optimization and the other is for QoE training (see Fig.\ref{fig:closeloop} and Fig.\ref{fig:interactions});

\begin{figure}[h!]
\begin{center}
\includegraphics[width=6cm]{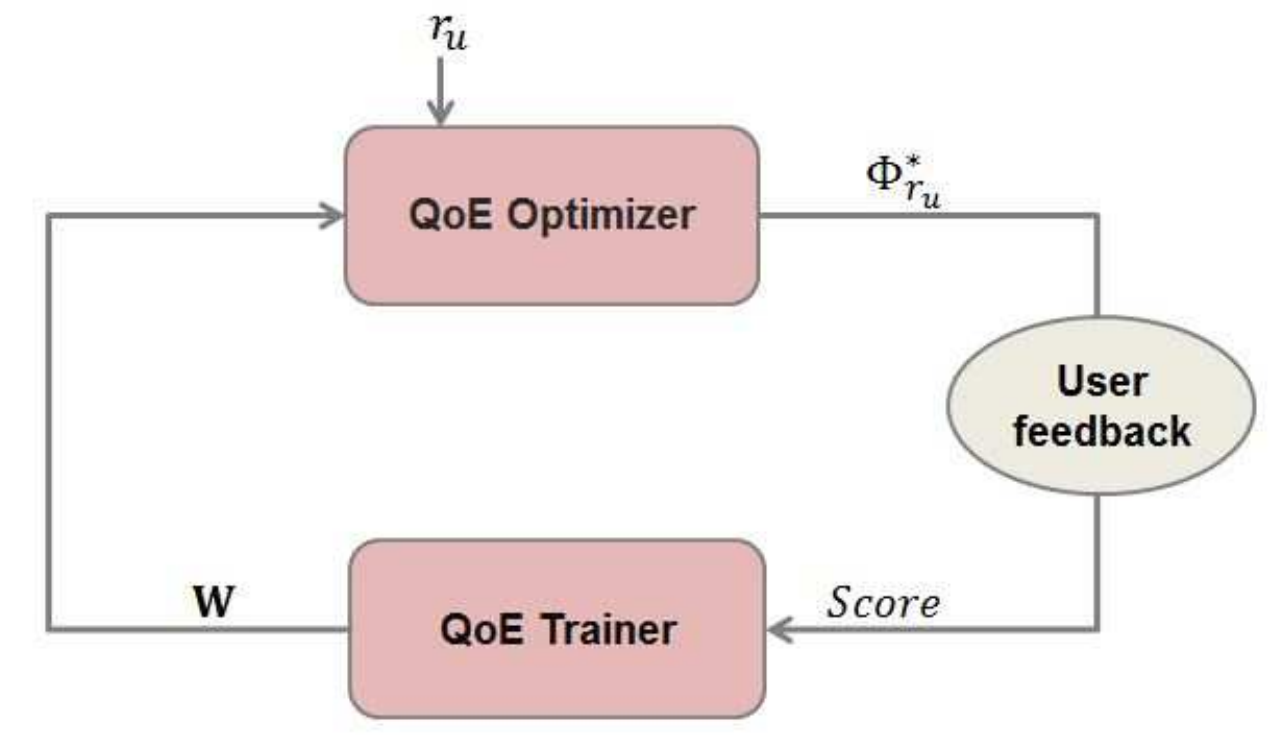}
\caption{Closed-loop based framework for QoE-optimization.}
\label{fig:closeloop}
\end{center}
\end{figure}

\begin{figure} [h!]
\begin{center}
\includegraphics[width=9cm]{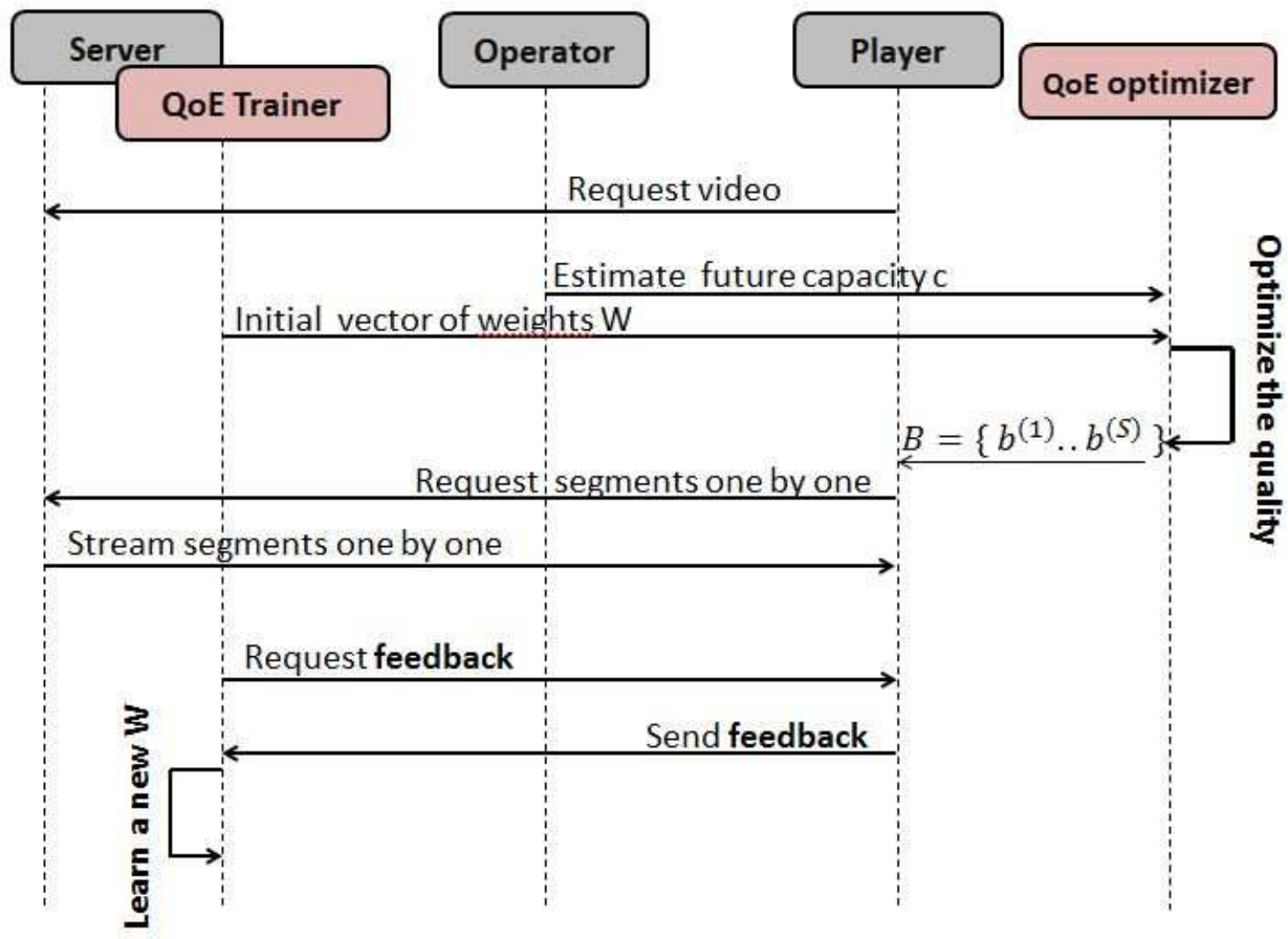}
\caption{Framework interaction with video streaming entities.}
\label{fig:interactions}
\end{center}

\end{figure}

\subsubsection{QoE training tool}
To compute $\mathcal{W}^*$, we use a simple neural network \cite{ MAGOULAS06IJBC}, where the training samples are couples of QoE metrics and user feedback.
We define the training dataset as $\{ ({\Phi^*}_{r_u}, {\mathcal{F}}_{r_u}) \}_{1 \le u \le U}$, where ${\Phi^*}_{r_u}$ is the vector of QoE metrics delivered by (\ref{max}) under throughput $r_u$ and vector $\mathcal{W}$. ${\mathcal{F}}_{r_u}$ being the corresponding feedback.

We define the activation function of the neural network as a linear function $h_\mathcal{W}(\Phi)=\mathcal{W}^\top \Phi$, where $\Phi$ is the input vector and $\mathcal{W}$ is the vector of weights to learn (See Fig. \ref{fig:NN}).

\begin{figure} [h!]
\vspace{-0.3cm}
\begin{center}
\includegraphics[width=5cm]{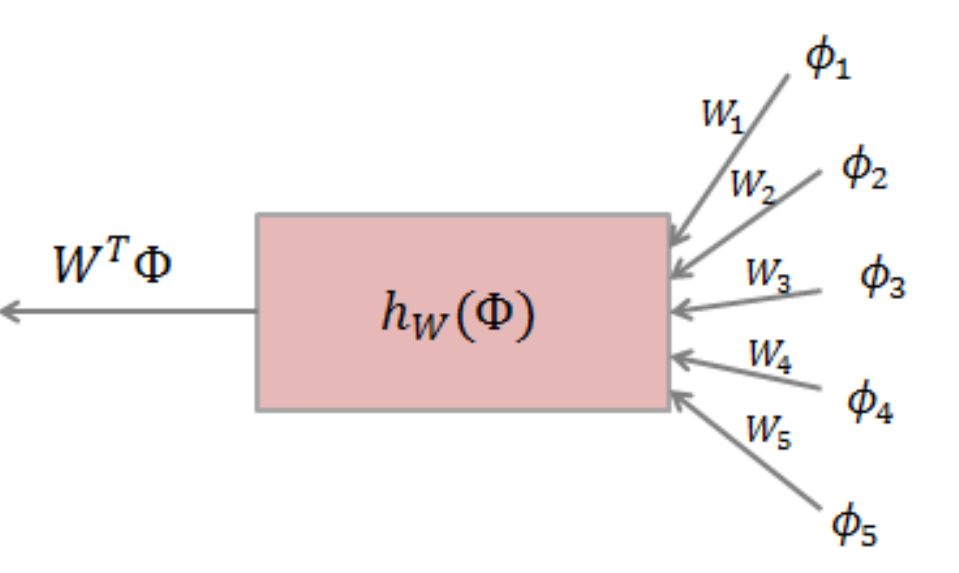}
\caption{Architecture of the QoE trainer.}
\label{fig:NN}
\end{center}
\end{figure}
We make use of a mini-batch learning algorithm based on the gradient descent. The goal behind using the gradient descent is to minimize the average error rate between $\mathcal{F}_{r_u}$ and the network output $h_{\mathcal{W}}({\Phi^*}_{r_u})$, $u \in \{1 \cdots U \}$.

Let $\textbf{Loss}(\mathcal{W}, {\Phi^*}_{r_u}, {\mathcal{F}}_{r_u})$ be the half squared error corresponding to the $u^{th}$ training sample
 and $\textbf{Loss}(\mathcal{W}, m)$ be the averaged error among $m$ training samples, namely
\begin{equation}
 \textbf{Loss}(\mathcal{W}, {\Phi^*}_{r_u}, {\mathcal{F}}_{r_u}) =\frac{1}{2}| h_{\mathcal{W}}({\Phi^*}_{r_u}) - {\mathcal{F}}_{r_u}|^2
\end{equation}
\vspace{-0.2cm}
\begin{equation}
\textbf{Loss}(\mathcal{W},m) =\frac{1}{m} \sum_{u=1}^m \textbf{Loss}(\mathcal{W}, {\Phi^*}_{r_u}, {\mathcal{F}}_{r_u})
\end{equation}
To reduce the average loss, the gradient descent
updates the vector of weights $\mathcal{W}$ in a way that it moves oppositely to the direction of the gradient vector $\nabla \textbf{Loss}(\mathcal{W},m)$. The algorithm stops when a predefined minimum loss $\epsilon$ is reached or when the number of updating steps is above a given threshold $T_{rs}$. (See Algorithm \ref{alg:gradient}).

The partial derivatives of $\textbf{Loss}(\mathcal{W},m)$ in function of the weights $\omega_k$, $k \le 5$ are given by
\begin{equation}
 \frac{\partial \textbf{Loss}(\mathcal{W},m)}{ \partial \omega_k} = \frac{\partial}{\partial \omega_k} \frac{1}{m} \sum_{u=1}^m \textbf{Loss}(\mathcal{W}, {\Phi^*}_{r_u}, {\mathcal{F}}_{r_u})
 \end{equation}

$$ = \frac{1}{m} \sum_{u=1}^m {\Phi^*}_{{r_u},k} (\mathcal{W}^\top {\Phi^*}_{r_u}- {\mathcal{F}}_{r_u})$$

\begin{algorithm}
\caption{The mini-batch Gradient descent}\label{alg:gradient}
\begin{algorithmic}[1]
\State \textbf{Input:$\{ ({\Phi_{r_1}}^*, \mathcal{F}_{r_1}), .. ,({\Phi_{r_m}}^*, \mathcal{F}_{r_m}) \}, \ \epsilon , \ \mu , \ T_{rs} , \ [\alpha_{min}, \alpha_{max}]$}
\State \small{GoodConvergence =0 ; SlowConvergence =0 ; Divergence = 0;} Set $\alpha$ in $[\alpha_{min}, \alpha_{max}]$ ; Set $\mathcal{W}$ very small;
\Repeat
\Repeat
\State $\mathcal{W} = \mathcal{W} - \alpha. \nabla \textbf{Loss}(\mathcal{W},m)$
\Until $\textbf{Loss}(\mathcal{W},m) \le \epsilon $ or $T_{rs}$ iterations are done
\If {$\textbf{Loss}(\mathcal{W},m) \le \epsilon$}
\State GoodConvergence =1
\Else
\If {$\textbf{Loss}(\mathcal{W},m)$ is decreasing }
\State SlowConvergence = 1
\State increase ($\alpha_{min}$)
\State Set $\alpha$ in $[\alpha_{min}, \alpha_{max}]$
\Else
\State Divergence =1
\State decrease($\alpha_{max}$)
\State Set $\alpha$ in $[\alpha_{min}, \alpha_{max}]$
\EndIf
\EndIf
\Until GoodConvergence or $(\alpha_{max}- \alpha_{min}) \le \mu $
\State $\textbf{return} \ \mathcal{W}^*= \mathcal{W}$
\end{algorithmic}
\end{algorithm}

\section {Numerical results} \label{section4}
\label{Results}
\subsection{Simulation environment}
We evaluate the performance of the proposed framework through extensive simulations using NS3 and Matlab.
NS3 was used to generate standard-compliant correlated throughputs. To get many throughput samples, we performed extensive simulations of an LTE network by varying the mobility of users each time. We put all NS3 parameter settings in Table \ref{parametersns3}.
The QoE optimization sub-framework and the QoE trainer were both developed using Matlab. 
As in real world, we consider users' feedbacks as scores rated from 1 to 5. When a quality ${\Phi_{r}}^*$ is delivered to a user, we look through the predefined synthetic QoE dataset to find the score it may give. In the dataset, we put all the possible values of vector ${\Phi_{r}}^*$ in a specific priority order, i.e., $|w_{satll}| >> |w_{rebuffering}| >> |w_{average-quality} | >> |w_{startup}| >> |w_{switching}|$. These vectors were then grouped in classes. To each class we associated a MOS and a specific distribution of scores. When ${\Phi_{r}}^*$ is delivered, we determine the class to which it belongs. Then, according to that class we randomly generate a score based on the distribution of scores in the dataset. 
Note that the throughput samples used at the level of the QoE optimization sub-framework were randomly selected (according to a Uniform distribution) among $1000$ throughput samples generated with NS3.
All Matlab parameter settings are listed in Table \ref{parameters}.

\begin{table}[t]
\footnotesize
\begin{center}
\begin{tabular}{|l|l|}
\hline
\textbf{Number of macro cells} & 1 \\
\hline
\textbf{Number of UEs per cell} & 10 \\
\hline
\textbf{eNb Tx Power} & 46 dBm \\
\hline
\textbf{eNb noise figure} & 5 dB \\
\hline
\textbf{UE noise figure} & 9 dB \\
\hline
\textbf{Pathloss model} & COST 231 \\
\hline
\textbf{MAC scheduler} & Proportional fair 50 RBs \\
\hline
\textbf{Fading model} & Pedestrian \\
\hline
\textbf{Transmission model} & MIMO Transmit diversity \\
\hline
\textbf{Mobility model} & RandomWalk2dMobilityModel \\
\hline
\textbf{Velocity of users} & Uniform [5,16] m/s \\
\hline
\textbf{EPS bearer} & NGBR-VIDEO-TCP-DEFAULT \\
\hline
\textbf{Fading model} & Pedestrian \\
\hline
\textbf{Simulation length} & 70 s \\
\hline
\end{tabular}
	\end{center}
	\caption{NS3 simulation setting parameters.}
	\label{parametersns3}
	\vspace{-0.5cm}
\end{table}

\begin{table}[t]
\footnotesize
\begin{center}
\begin{tabular}{|l|l|}
\hline
\textbf{Window Size} & 70 s \\
\hline
\textbf{Throughput Time Slot} & 1 s \\
\hline
\textbf{Video Length} & 30 s \\
\hline
\textbf{Segment Length} & 1s \\
\hline
\textbf{Video frame rate} & 30 fps \\
\hline
\textbf{Playback cache} & 5s \\
\hline
\textbf{Bit-rate levels Mbps} & [0.4 0.75 1 2.5 4.5] \\
\hline
\textbf{Maximum number of stalls (p)} & 1 \\
\hline
\end{tabular}
	\end{center}
	\caption{Matlab simulation setting parameters.}
	\label{parameters}
	\vspace{-0.5cm}
\end{table}

\begin{figure} [h!]
\begin{center}
\includegraphics[width=7cm]{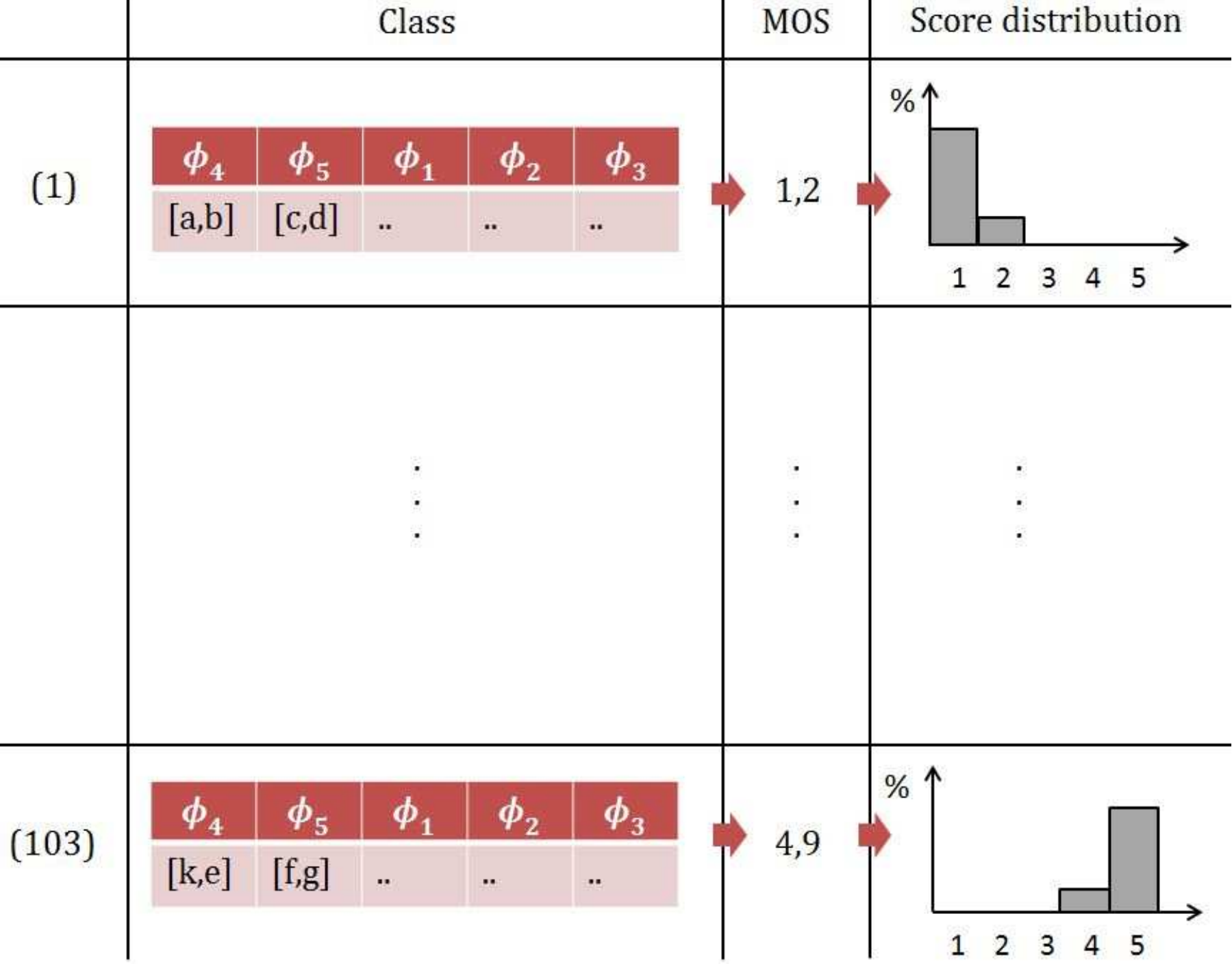}
\caption{Synthetic-dataset for score generation.}
\label{fig:dataset}
\end{center}
\end{figure}

\begin{figure} [h!]
\begin{center}
\includegraphics[width=8cm]{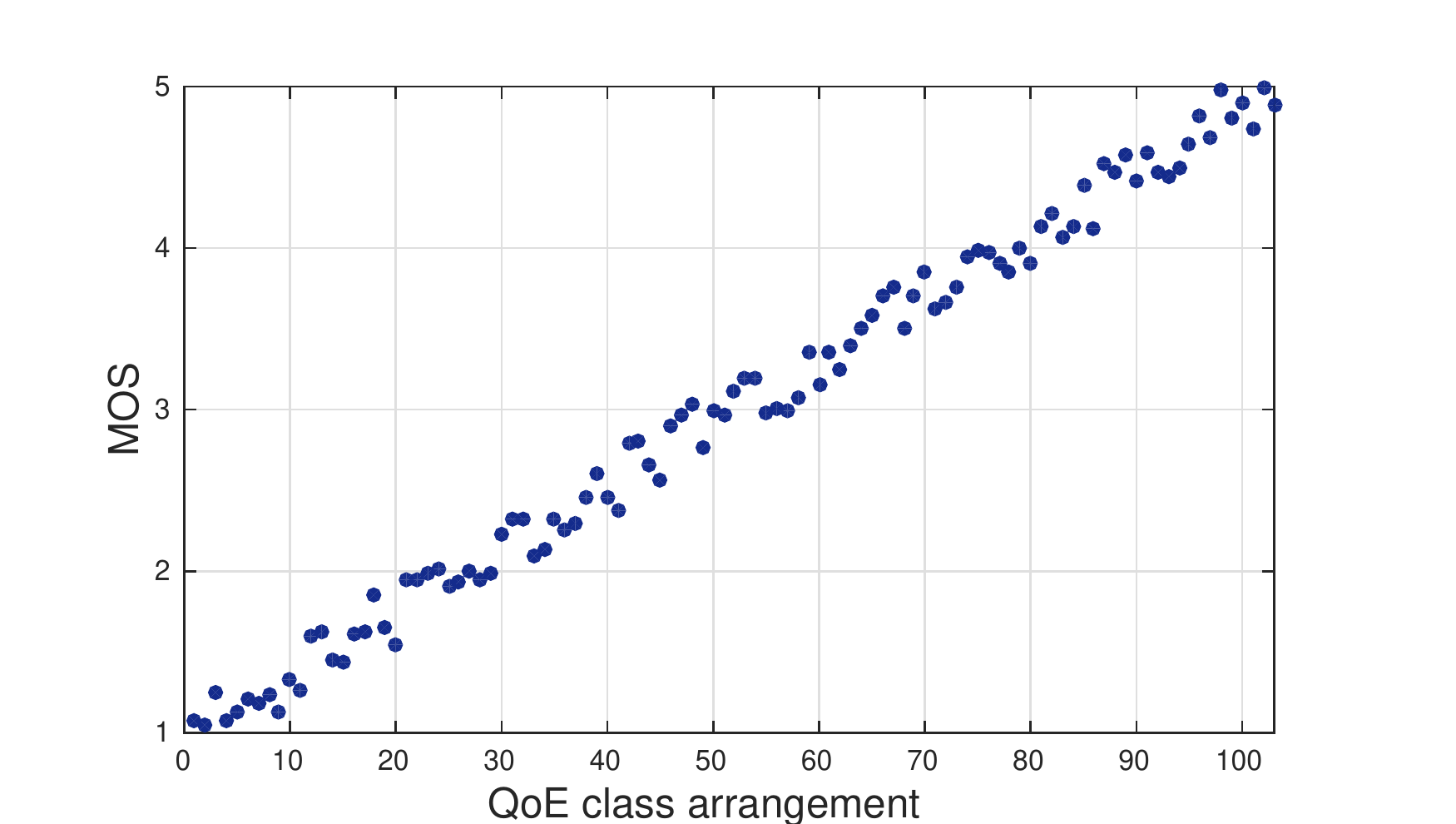}
\caption{Synthetic MOS as function of the QoE class arrangement.}
\label{fig:distribution}
\end{center}
\end{figure}

\subsection{Performance results}
The performance evaluation of the closed-loop based framework allows us to show that (i) the learning converges ultimately to a steady state, in which the learning output is a quasi-constant vector $\mathcal{W}^*$, and that (ii), more importantly, this vector $\mathcal{W}^*$ achieves the highest QoE compared to the other vectors computed throughout the learning process.

In Fig. \ref{fig:W}, we show the evolution of the mean square variation of vector $\mathcal{W}$ during the learning process for different values of the mini-batch size. Results show that for all cases, the variation tends to zero, although the decrease is slow  in some cases (case of 5 and 50 scores). A fast convergence is however noticed in the case of 10 scores. The difference in the convergence time is actually due to the random character of the throughput selection and the scores generation.
In a second step, we were comparing the final outputs $\mathcal{W}^*$. We noticed that they were not exactly the same. Hence, we computed the MOS when each of the previous updated values of $\mathcal{W}$ was applied with the QoE optimization sub-framework under $1000$ randomly selected throughputs. Fig. \ref{fig:MOS} shows that for the four mini-batch sizes, the MOS experiences some fluctuations with the first values of $\mathcal{W}$. Then, when it tends to the values obtained at the steady state, it converges to the highest MOS value (around 4.8 for the four cases). These results offer hope that the proposed closed-loop based framework can be designed around QoE optimization for video adaptation and delivery in real-world environment.

\begin{figure} [h]
\begin{center}
\vspace{0cm}
\includegraphics[width=8.5cm]{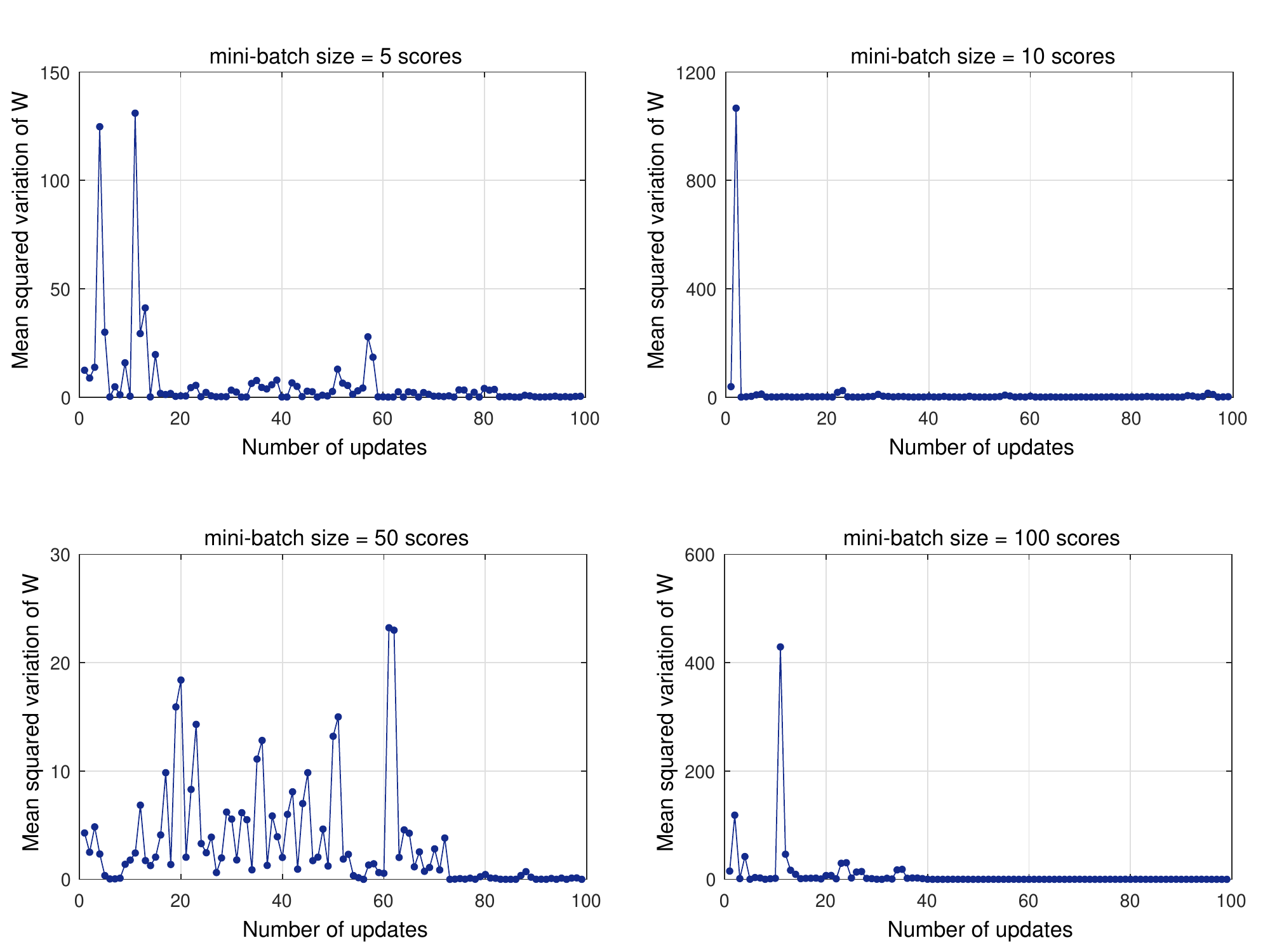}
\caption{The mean square variation of vector $\mathcal{W}$ during the learning process.}
\label{fig:W}
\end{center}

\begin{center}
\includegraphics[width=8.5cm]{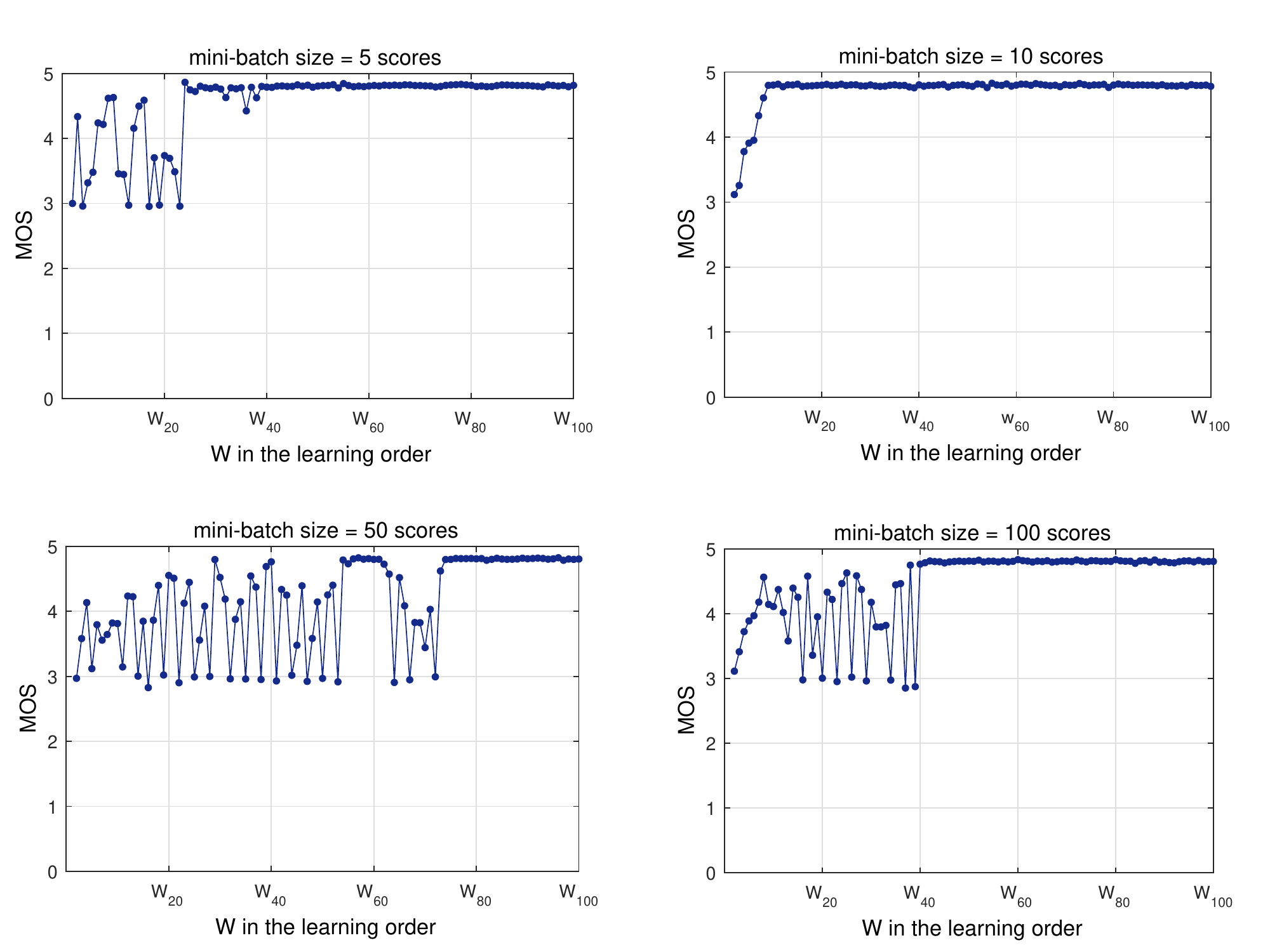}
\caption{The MOS of the QoE-optimization sub-framework using the updated values of vector $\mathcal{W}$.}
\label{fig:MOS}
\end{center}

\end{figure}

\section{conclusion} \label{conclusion}
In this paper, we have addressed a QoE optimization problem with machine learning to optimize the quality of the delivered video by fitting the real profiles of the users. We have proposed a closed-loop framework based on the users' feedbacks to learn their corresponding QoE function and to proceed to their QoE optimization. By using a synthetic QoE dataset, we have shown the efficiency of the proposed closed-loop system. Indeed, the QoE function learned at the steady state ensures a high quality delivery for the majority of users. These promising results allow us to gain insight on how QoE optimization problem can be handled in a heterogeneous population. As a future step, real scores on real video streaming will be collected in order to study the robustness of the proposed solution.

\section*{Acknowledgement}
This research is partially supported by NSF grants CNS-1720230, CNS-1544782, and SES-1541164.

\bibliographystyle{ieeetr}
\bibliography{mybib}

\begin{thebibliography}{10}

\bibitem{CiscoSurvey}
{\em {Cisco Visual Networking Index: Forecast and Methodology, 2015-2020}}.
\newblock
  {http://www.cisco.com/c/en/us/solutions/collateral/service-provider/visual-networking-index-vni/complete-white-paper-c11-481360.pdf}.

\bibitem{Miguel16NPA}
M.~G. Pineda, S.~Felici{-}Castell, and J.~Segura{-}Garcia, ``Using factor
  analysis techniques to find out objective video quality metrics for live
  video streaming over cloud mobile media services,'' {\em Network Protocols
  {\&} Algorithms}, 2016.

\bibitem{Mellouk12NOC}
M.~S. Mushtaq, B.~Augustin, and A.~Mellouk, ``Empirical study based on machine
  learning approach to assess the qos/qoe correlation,'' in {\em Networks and
  Optical Communications (NOC), 2012 17th European Conference on}, 2012.

\bibitem{Vriendt14BellLabs}
J.~D. Vriendt, D.~D. Vleeschauwer, and D.~C. Robinson, ``{QoE model for video
  delivered over an LTE network using HTTP adaptive streaming},'' {\em Bell
  Labs Technical Journal}, 2014.

\bibitem{Balachandran12ACM}
A.~Balachandran, V.~Sekar, A.~Akella, S.~Seshan, I.~Stoica, and H.~Zhang, ``A
  quest for an internet video quality-of-experience metric,'' in {\em
  Proceedings of the 11th ACM Workshop on Hot Topics in Networks}, ACM, 2012.

\bibitem{Amour2015}
L.~Amour, S.~Souihi, S.~Hoceini, and A.~Mellouk, ``{A Hierarchical
  Classification Model of QoE Influence Factors},'' {\em {13th International
  Conference Wired/Wireless Internet Communications (WWIC) Revised Selected
  Papers}}, pp.~225--238, 2015.

\bibitem{DOB11}
F.~Dobrian, V.~Sekar, A.~Awan, I.~Stoica, D.~Joseph, A.~Ganjam, J.~Zhan, and
  H.~Zhang, ``Understanding the impact of video quality on user engagement,''
  {\em ACM SIGCOMM Computer Communication Review}, vol.~41, no.~4,
  pp.~362--373, 2011.

\bibitem{videoEngagement}
{\em {Youtube:} {Measure} {video ad performance}}.
\newblock {https://support.google.com/youtube/answer/2375431}.

\bibitem{Balachandran13ACM}
A.~Balachandran, V.~Sekar, A.~Akella, S.~Seshan, I.~Stoica, and H.~Zhang,
  ``Developing a predictive model of quality of experience for internet
  video,'' {\em SIGCOMM Comput. Commun. Rev.}, 2013.

\bibitem{Testolin14MED-HOC-NET}
A.~Testolin, M.~Zanforlin, M.~D. F.~D. Grazia, D.~Munaretto, A.~Zanella,
  M.~Zorzi, and M.~Zorzi, ``A machine learning approach to qoe-based video
  admission control and resource allocation in wireless systems,'' in {\em Ad
  Hoc Networking Workshop (MED-HOC-NET), 2014 13th Annual Mediterranean}, 2014.

\bibitem{Bampis}
{Christos G. Bampis, and Alan C. Bovik}, ``Learning to predict streaming video
  qoe: Distortions, rebuffering and memory,'' 2017.
\newblock Available at https://arxiv.org/abs/1703.00633.

\bibitem{she2010}
U.~Shevade, Y.-C. Chen, L.~Qiu, Y.~Zhang, V.~Chandar, M.~K. Han, H.~H. Song,
  and Y.~Seung, ``Enabling high-bandwidth vehicular content distribution,'' in
  {\em ACM CoNEXT, Philadelphia, USA.}, 2010.

\bibitem{KnowingFutureInfocom13}
Z.~Lu and G.~de~Veciana, ``Optimizing stored video delivery for mobile
  networks: The value of knowing the future,'' in {\em INFOCOM, 2013
  Proceedings IEEE}, pp.~2706--2714, April 2013.

\bibitem{ImenWoWMOM16}
I.~Triki, R.~El-Azouzi, and M.~Haddad, ``{NEWCAST: Anticipating Resource
  Management and QoE Provisioning for Mobile Video Streaming},'' in {\em 17th
  International Symposium on a World of Wireless, Mobile and Multimedia
  Networks (IEEE WoWMoM)}, June 2016.

\bibitem{ism/TrikiAH16}
I.~Triki, R.~E. Azouzi, and M.~Haddad, ``{Anticipating Resource Management and
  QoE for Mobile Video Streaming under Imperfect Prediction},'' in {\em {IEEE}
  International Symposium on Multimedia, {ISM}, San Jose, CA, USA, December
  11-13, 2016}, pp.~93--98, 2016.

\bibitem{Yin:2015:CAD}
X.~Yin, A.~Jindal, V.~Sekar, and B.~Sinopoli, ``{A Control-Theoretic Approach
  for Dynamic Adaptive Video Streaming over HTTP},'' {\em SIGCOMM Comput.
  Commun. Rev.}, pp.~325--338, 2015.

\bibitem{MAGOULAS06IJBC}
G.~D. MAGOULAS and M.~N. VRAHATIS, ``Adaptive algorithms for neural network
  supervised learning: A deterministic optimization approach,'' {\em
  International Journal of Bifurcation and Chaos}, 2006.

\end{thebibliography}

\end{document}